\documentclass{fundam}

\usepackage{amssymb,amscd,mathptmx,MnSymbol,stackengine}
 \usepackage{url,chngcntr,amsfonts,qtree,synttree}
\usepackage{enumitem}

\newcommand{\espa}{\makebox[1cm]{}}
\newcommand{\esp}{\makebox[.5cm]{}}

\newcommand{\ssi}{\makebox[1cm]{iff}}

\newcommand{\nl}{\vspace{1mm} \\}

\newcommand{\n}{\mathbb N}
\newcommand{\reel}{\mathbb R}
\newcommand{\words}{\{0,1\}^\ast}
\newcommand{\nwords}{\n^{\ast}}
\newcommand{\cantor}{\{0,1\}^{\n}}
\newcommand{\baire}{\n^{\n}}

\usepackage{hyperref}

\begin{document}

\setcounter{page}{345}
\publyear{2021}
\papernumber{2077}
\volume{182}
\issue{4}

  \finalVersionForARXIV

\title{The Inverse of Ackermann Function is Computable in Linear Time}

\author{Claude Sureson\thanks{Address  for correspondence: Universit\'e Paris 7 Denis Diderot,
                          France.\newline \newline
          \vspace*{-6mm}{\scriptsize{Received  September 2021; \ revised September 2021.}}}
\\
Universit\'e Paris 7 Denis Diderot\\
5 Rue Thomas Mann, 75013 Paris, France\\
sureson@math.univ-paris-diderot.fr
}

\maketitle

\runninghead{C. Sureson}{The Inverse of Ackermann Function is Computable in Linear Time}

\begin{abstract}
We propose a detailed proof of the fact that the inverse of Ackermann function is computable in linear time.
\end{abstract}


\section{Introduction}

The Ackermann function was proposed in 1926 by W.\,Ackermann~(see \cite{ack}) as a simple example of a total recursive function which is not primitive recursive. It is often presented, as done initially by R.\,P\'eter, under the form of  a two argument function  \,$A:\n\times\n\to\n$.

The function \,$n\mapsto A(n,n)$\, grows extremely fast (asymptotically faster than any primitive recursive function). Hence its inverse, denoted \,$\alpha$, grows very slowly; it is known to be primitive recursive. The function \,$\alpha$\, appears to express time complexities in data structure analysis as in the work of E.\,Tarjan~\cite{tar} and in algorithmic geometry as in the work of B.\,Chazelles~\cite{cha}. It is also used by G.\,Nivasch, R.\,Seidel and M.\,Sharir without reference to the original Ackermann function $A$ in~\cite{niva,seidel2,seidel}.

In a previous work~\cite{sur}, we needed a bound on the amount of time spent to compute the function \,$\alpha$. But except for the fact that \,$\alpha$\, is primitive recursive, we could not find a documented reference. This is why we proposed a detailed proof of the fact that \,$\alpha$\, is computable in linear time (on a multitape Turing machine). Once our work was made public, L.\,Tran, A.\,Mohan and A.\,Hobor~\cite{tmh} informed us that they had obtained  a similar result by totally different methods (functional programming techniques). Our demonstration is elementary and builds partly on the exposition by G.\,Tourlakis~\cite{tou} of the primitive recursiveness of the graph of $A$.

\section{A few classical definitions}

\subsection{Some notation}

$\n,\,\mathbb Z$ and $\reel$ represent respectively the set of natural, integer  and real numbers.
$\words$ and $\cantor$ denote the sets of finite and infinite binary sequences.
$\nwords$ and $\baire$ are the sets of finite and infinite sequences of natural numbers.

\begin{definition}~\label{bidon}
\begin{enumerate}
\item Let \,$\mathbf x$\, be a finite or infinite sequence. For an integer \,$i\in\n,\ \,\mathbf x\restriction_ i$ is the restriction of \,$\mathbf x$\, onto the set \,$\{0,1,\ldots,i-1\}$.
\item If \,$\mathbf x$\, is a finite sequence, then \,$|\mathbf x|$ \,denotes its length.
\item $<_{lex}$ \ is the lexicographic order on \,$\nwords$.
\item Let $n\in\n$. Then \,$|n|$\, is the length of the string $\sigma_n$ corresponding to $n$ under binary representation.
For $n\geq 1$,  \,$|n|\,=\,\lfloor\log_2 n\rfloor+1$ \,($|0|=1$) and \ $n<2^{|n|}\leq 2n$.
\item Let \ $log\,:\n\setminus\{0\} \to\n$ \ be defined, for $n\geq 1$ by \ $log(n)=\lceil\log_2(n)\rceil$. \ Then \ $|n|-1\leq log(n)\leq|n|$.
\end{enumerate}
\end{definition}

All complexity notions refer to binary representation of integers. Given a function $f:\n\to\n$ which is time constructible (see~\cite{arora}) and such that \,$f(n)\geq n$ \,for all \,$n\in\n$, we shall consider predicates checkable in time \ $\+O(f(n))$ and functions computable in time \,$\+O(f(n))$.

\subsection{Definition of the Ackermann function}

There exist different versions of Ackermann function depending on the initial definitions (\emph{i.e.} the values of \,$A(0,n)$\, and of \,$A(k,0)$, \,for \,$k,n\in\n$). We refer to the definition in~\cite{co-la} and freely use the properties proved in this textbook.\medskip

\begin{definition}[{\cite[5.2.1]{co-la}}]\label{def-A}
\begin{enumerate}[label=(\alph*),topsep=-1mm,itemsep=-.2ex]
\item
Let \,$A:\n\times\n\to\n$\, be defined as follows: for \,$k,n\in\n$,
\begin{itemize}[topsep=-1mm,itemsep=-.2ex]
\item $A(0,n)=2^n$,
\item $A(k,0)=1$,
\item $A(k+1,n+1)=A(k,A(k+1,n))$.
\end{itemize}
\item For \,$k\in\n$, \,let \,$A_k:\n\to\n$\, be such that for all \,$n\in\n$, \,$A_k(n)=A(k,n)$.
\item Let \,$Ack:\n\to\n$\, be such that \,$Ack(n)=A(n,n)$.
\end{enumerate}
\end{definition}

We chose this version rather than Tourlakis' one because it allows some simplifications and because it is closely related to the version \,$A^T$\, proposed by Tarjan~\cite{tar} and refered to in~\cite{cha}. $A^T$ is defined as follows:
\begin{itemize}[topsep=-1mm,itemsep=-.2ex]
\item for $n\in\n,\ \,A^T(0,n)=2n$,
\item for $k\in\n,\ \,A^T(k,0)=0$ \,and\, $A^T(k,1)=2$,
\item for $k\in\n,\,n\geq 1, \,A^T(k+1,n+1)=A^T(k,A^T(k+1,n))$.
\end{itemize}

One can check that for any \,$k\in\n,\,n\geq 1,\ \,A^T(k+1,n)=A(k,n)$.\nl
We recall the notion of inverse. The methods developed in this paper can be applied to inverse functions with two parameters (see~\cite{cha,tar}), but we shall not consider them here. We should also mention the work of~\cite{seidel,seidel2} using ``inverse Ackermann functions" without refering explicitely to the Ackermann function itself.

\begin{definition}~\label{def-ack}
\begin{itemize}[topsep=-1mm,itemsep=-.2ex]
\item Let \,$f:\n\to\n$\, be unbounded and nondecreasing. The  inverse of $f$ denoted \,$Inv_f$\, is defined as follows: for any \,$n\in\n,\ \,Inv_f(n)$ \,is the least $k\in\n$ such that \,$f(k)\geq n$.
\item Let \,$\alpha:\n\to\n$ \,be\, $Inv_{Ack}$.
\end{itemize}
\end{definition}

Because of the above relation between $A$ and $A^T$,
results about $A,\ Ack$ and $\alpha$ can thus be applied to $A^T$ and its related ``inverses".\nl
We recall some basic properties of the functions \,$A_k$, \,for $k\in\n$:

\begin{lemma} \label{property_Ak} For any \,$k\in\n$,
\begin{enumerate}[label=(\alph*),topsep=-1mm,itemsep=-.2ex]
\item $A_k$\, is strictly increasing (see \cite[lemma 5.7]{co-la}),\vspace*{-0.5mm}
\item for any \,$n\geq 1$, \,$A_k(n)\leq A_{k+1}(n)$ (see \cite[lemma 5.8]{co-la}),\vspace*{-0.5mm}
\end{enumerate}
\end{lemma}

We shall consider iterates of a function:

\begin{definition} Given a function \,$g:\n\to\n$\, and $\,m\in\n$, \,the $m$th iterate of $g$, denoted $g^{(m)}$ is defined inductively by: \,$g^{(0)}(n)=n$\, and \,$g^{(m+1)}(n)=g(g^{(m)}(n))$.
\end{definition}

To simplify notation (avoiding  towers of exponentials), we shall apply the notion to the following function: \medskip

\begin{definition} Let $exp:\n\to\n$ be such that $exp(n)=2^n$, \,for \,$n\in\n$.
\end{definition}\vspace*{-1mm}

\section{Properties of the functions \,$\pmb {A_k}$\, and of their inverses}

We first note some elementary properties of $A$:

\begin{fact}\label{fact p3} For any \,$i\in\n$,
\begin{enumerate}[label=(\arabic*),topsep=-1mm,itemsep=-.2ex]
\item $A(i,1)=2$,
\item $A(i,2)=4$,
\item $A(1,i)=exp^{(i)}(1)$.
\end{enumerate}
\end{fact}

\eject

\begin{proof}

\vspace*{-6mm}
\begin{enumerate}
\item $A(0,1)=2$ \,by definition, and for any $i\in\n,\\\ \,A(i+1,1)=A(i,A(i+1,0))=A(i,1)$.
\item $A(0,2)=2^2=4$ \,by definition, and for any $i\in\n,\\ \,A(i+1,2)=A(i,A(i+1,1))=A(i,2)$ by (1).
\item $A(1,0)=1=exp^{(0)}(1)$  \,by definition, and for any $i\in\n,\\ \,A(1,i+1)=A(0,A(1,i))=exp(A(1,i))$.
\end{enumerate}

\vspace*{-6mm}
\end{proof}

The link between $A_k$ and $A_{k+1}$ is the following one:

\begin{fact}\label{fact p.4 TSVP}
 For any $k,\,n\in\n,\ \,A_{k+1}(n)=A_k^{(n)}(1)$.
\end{fact}

\begin{proof}
Let $k$ be fixed. This is true for $n=0$: $A_{k+1}(0)=1=A_k^{(0)}(1)$.\\
Let us assume the equality holds for \,$n\in\n$. Then\nl
\espa$A_{k+1}(n+1)=A_k(A_{k+1}(n))=A_k(A_k^{(n)}(1))=A_k^{(n+1)}(1)$.
\end{proof}

We deduce from Fact~\ref{fact p3}, some lower bounds:

\begin{fact}\label{case3} $A_3(3)>exp^{(4)}(3)$.
\end{fact}

\begin{proof}

\vspace*{-4mm}
$\begin{array}[t]{cclr}
A_3(3)&=&  A_2(A_3(2))&\\
      &=& A_2(4)  &(\text{by~\ref{fact p3}(2))}\\
      &=& A_1(A_2(3))&\\
      &=& A_1(A_1(A_2(2)))&\\
      &=&A_1(A_1(4))  &(\text{by~\ref{fact p3}(2))}\\
      &=& A_1(exp^{(4)}(1)) &(\text{by~\ref{fact p3}(3))}\\
      &=&A_1(2^{16})&\\
      &=&exp^{(2^{16}-2)}(4)&(\text{by~\ref{fact p3}(3))}\\
      &>&  exp^{(4)}(3).&
\end{array}$

\vspace*{-5mm}
\end{proof}

\smallskip
\begin{claim}\label{claim p4}
For any $n\geq 3,\ \,A_3(n)>exp^{(4)}(n)$.
\end{claim}

\begin{proof} By the previous fact, this is true for \,$n=3$.\\
We thus argue by induction, assuming the inequality holds for \,$n\geq 3$. Then
\begin{align*}
A_3(n+1)=A_2(A_3(n))>A_2(exp^{(4)}(n))&\geq\, A_0(exp^{(4)}(n))\\&\geq\, exp^{(5)}(n)=exp^{(4)}(2^n)\geq exp^{(4)}(n+1).
\end{align*}

\vspace*{-7mm}
\end{proof}

\begin{claim}\label{claim1 p5}
For any $k\geq 3,\ \,A_k(3)>exp^{(4)}(k)$.
\end{claim}

\begin{proof} By~\ref{case3}, this holds for $k=3$.\nl
We assume \,$A_k(3)>exp^{(4)}(k)$ \,for $k\geq 3$. Then\nl
$A_{k+1}(3)=A_k(A_{k+1}(2)) \underset{\text{\eqref{fact p3}}}{=}A_k(4)=A_{k-1}(A_k(3))\underset{\text{ind.}}{>}A_0(exp^{(4)}(k)=exp^{(4)}(2^k)
\geq\ exp^{(4)}(k+1)$.
\end{proof}

We now evaluate the complexity of the functions $Inv_{A_k}$, for $k\in\n$.

\begin{lemma}\label{lemma2 p5}
$log$ \,is computable in linear time.
\end{lemma}

\begin{proof}
This is folklore. We propose a simple argument suggested by one referee:
one counts in binary the number of digits of the input x.\\ The counter being written in reverse order, we change the first digit of the counter for all browsed positions on the input tape, change the second digit for every position out of 2,..., change the $k$th digit for every position $r$ on the input tape such that $r-1$ has binary representation of the form $u1^{k-1}$...\\
Hence for some constant B, if \,$2^r\leq |x|<2^{r+1}$, the number of steps required to obtain $|x|$ in binary is bounded by \ $B(|x|+\sum_{k=1}^{k=r}2^{r-(k-1)})\,=\,O(|x|)$.
\end{proof}


One notes that \,$Inv_{A_0}=\lceil \log_2\rceil=log$. \ Hence we can state:

\begin{claim}\label{inv a0}
$Inv_{A_0}$ is computable in linear time.
\end{claim}
We now relate \,$Inv_{A_{k+1}}$\, to \,$Inv_{A_k}$, for $k\in\n$.

\begin{definition} \label{def nr}
Let $m,k\geq 0$ and let the sequence of integers \,$(n_r)_{r\leq s}$\, be defined inductively as follows:
\begin{itemize}[topsep=-1.5mm,itemsep=-.4ex]
\item $n_0=m$,
\item for \,$r\geq 0$\, and \,$n_r$\, defined,
\begin{itemize}[topsep=-1.5mm,itemsep=-.8ex]
\item if $n_r\leq 1$, then we stop the construction and set $s=r$,
\item otherwise let \,$n_{r+1}=Inv_{A_k}(n_r)$.
\end{itemize}
\end{itemize}
\end{definition}

\begin{claim}\label{nr stops} Let \,$m,k,s$\, be as in the above definition.
\begin{enumerate}[label=(\alph*),topsep=-1mm,itemsep=-.4ex]
\item The construction does stop.
\item $Inv_{A_{k+1}}(m)=s$.
\end{enumerate}
\end{claim}

\begin{proof}
(a) \ If $m\leq 1$, then the construction stops at the first step and \,$s=0$.\nl
Hence let \,$m>1$. \,We check that the sequence \,$(n_r)_r$\, is strictly decreasing.\\
By definition, as long as \,$n_{r+1}$\, is defined, \ \begin{equation}\label{7.1}
A_k(n_{r+1}-1)\,<\,n_r\,\leq \,A_k(n_{r+1})
\end{equation}
Hence \ \ $2^{n_{r+1}-1}\, =\, A_0(n_{r+1}-1)\, \leq\, A_k(n_{r+1}-1)\, <\, n_r$. \ This gives \ $2^{n_{r+1}}\,<\,2n_r$. Since for any \,$t\in\n$, \,$2^t\geq 2t$, we deduce \ $n_{r+1}<n_r$.

\medskip
\noindent (b) \ If \,$m\leq 1$, \,then \ $A_{k+1}(0)\geq m$. \ Hence \ $Inv_{A_{k+1}}(m)=0=s$.\nl
Otherwise \,$s\geq 1$\, and we verify both inequalities: \ $A_{k+1}(s)\geq m$ \,and \, $A_{k+1}(s-1)<m$.

\medskip
${\pmb{A_{k+1}(s)\geq m}}$\,:\\
We check by induction on \,$t\leq s$\, that \,$n_{s-t}\leq A_k^{(t)}(1)$.\nl
- This is true for $t=0$ since \,$n_s\leq 1$.\\
- We assume this holds for \,$t\geq 0$. By~\eqref{7.1}, we deduce
\[n_{s-(t+1)}\,\leq\,A_k(n_{s-t})\,\underset{\text{ind.}}{\leq}\,A_k(A_k^{(t)}(1))\,\leq\,A_k^{(t+1)}(1).\]
By applying the inequality to \,$t=s$, we obtain from Fact~\ref{fact p.4 TSVP}\nl\centerline{ $m=n_0\,\leq\,A_k^{(s)}(1)\,=\,A_{k+1}(s)$ \,.}

\medskip
${\pmb{A_{k+1}(s-1)< m}}$\,:\\
We check by induction on \,$1\leq t\leq s$\, that \ $A_k^{(t-1)}(1)<n_{s-t}$.\nl
- Let \,$t=1$. \,then \,$n_{s-1}>1=A_k^{(0)}(1)$. \,Hence the inequality holds.\\
- We assume \,$n_{s-t}>A_k^{(t-1)}(1)$\, for \,$t\geq 1$. \ Hence \,$A_k^{(t-1)}(1)\leq n_{s-t}-1$.
  We deduce \nl \centerline{$A_k^{(t)}(1)=A_k(A_k^{(t-1)}(1))\leq A_k(n_{s-t}-1)\underset{\eqref{7.1}}{<} n_{s-(t+1)}$.}\nl
  Applying the inequality to \,$t=s$, we obtain \
  $m=n_0>A_k^{(s-1)}(1)\underset{Fact\ref{fact p.4 TSVP}}{=}A_{k+1}(s-1)$.\smallskip

We conclude that \,$Inv_{A_{k+1}}(m)=s$.
\end{proof}

From this characterization of \,$Inv_{A_{k+1}}$, we shall derive:

\begin{lemma}\label{time invk}
For any \,$k\in\n$, \,$Inv_{A_k}$ \,is computable in linear time.
\end{lemma}

\begin{proof} \ We shall argue by induction on \,$k\in\n$.\nl
- This holds for \,$k=0$\, by Claim~\ref{inv a0}.\nl
- We assume now that \,$Inv_{A_k}$\, is computable in linear time and we check that it is also the case for \,$Inv_{A_{k+1}}$.

Starting with \,$m\geq 2^{4}$,  we shall evaluate the time required to obtain the sequence \,$(n_r)_{r\leq s}$\, of Definition~\ref{def nr}. We recall that \ $n_{r+1}=Inv_{A_k}(n_r)$, \,if $n_r>1$.

Since \,$n_{s-1}>1$, \,$A_k(0)=1$, \,$n_s\leq 1$\, and \,$n_s=Inv_{A_k}(n_{s-1})$, \,necessarily \,$n_s=1$.

\begin{claim}\label{bound on s}
For \,$m\geq 4$, \,$s\,\leq\, 2\,log^{(2)}(m)$.
\end{claim}

\begin{proof} \ We note that since $A_k(1)=2<m=n_0$, necessarily $n_1>1$ and $s\geq 2$.\nl
By definition, $n_1$ and $n_2$ satisfy the following:\nl
\centerline{$2^{n_1-1}\leq A_k(n_1-1)<m$ \ and \ $2^{n_2-1}\leq A_k(n_2-1)<n_1$.}\nl
Now \,$2^{n_1}<2m$\, gives \,$n_1\leq\,log(m)$. Similarly we obtain \, $n_2\leq\,log(n_1)$ and hence
\begin{equation}\label{e1}
n_2\,\leq\,log^{(2)}(m)
\end{equation}
Since \,$(n_r)_{r\leq s}$\, is strictly decreasing and $n_s=1$, one checks that \,$n_2\,\geq\, 1+(s-2)$.\nl
Hence \ $s\,\leq\,n_2+1\,\underset{\eqref{e1}}{\leq}\,log^{(2)}(m)+1\,\leq\,2\,log^{(2)}(m)$ (the last inequality holds because $log^{(2)}(m)\geq 1$).
\end{proof}

Let \,$m\geq 2^4$. \,By induction hypothesis, there is a constant $C$ such that for any $u\in\n$, the computation of \,$Inv_{A_k}(u)$\, takes at most \,$C|u|$\, steps.\\
\esp - Hence the obtention of $n_1$ and $n_2$ takes at most \,$2C|m|$\, steps.\\
\esp - We now bound the time required to compute \,$(n_r)_{2<r\leq s}.$ \,For each \,$2\leq r<s$, the obtention of $n_{r+1}$ (given $n_r$) takes at most \,$C|n_2|$\, steps. Since \,$m\geq 2^4$, we have:\nl
\centerline{ $|n_2|\underset{\eqref{e1}}{\leq}|log^{(2)}(m)|\,\leq\,log^{(3)}(m)+1\,\leq\,2\,log^{(3)}(m)$,} \\
Hence to treat all \,$2\leq r<s$, by Claim~\ref{bound on s}, one needs at most \,$4C\,log^{(2)}(m)\,log^{(3)}(m)$\, steps.\\ There is a constant $D$ such that for any $m\in\n$, one has\nl
\centerline{ $log^{(2)}(m)\,log^{(3)}(m)\,\leq\,D\,log(m)\,\leq\,D\,|m|$.}\nl
Hence we deduce that $Inv_{A_{k+1}}(m)$ is computable in time \,$\+O(|m|)$.
\end{proof}

\medskip
There may be a way to use G.\,Nivasch (see~\cite{niva})  development on inverse Ackermann function to evaluate the complexity of the functions $Inv_{A_k}$. One would have to clarify the link between Nivasch's function $\alpha_k$ and our $Inv_{A_k}$.\vspace*{-1mm}

\section{Encoding sequences}

In this section, we introduce the coding of couples, triples or finite sequences of integers of arbitrary length.

\begin{definition}\label{def couple}~
\begin{itemize}[topsep=-1mm,itemsep=-.2ex]
\item For $u,\,v\in\n$, let \,$\displaystyle \langle u,v\rangle\,=\,\frac{(u+v)(u+v+1)}{2} +v$. \ Then \,$\langle \cdot,\cdot\rangle:\n\times\n\to\n$ \ is a bijection. \vspace*{-0.5mm}
\item Let $(\cdot)_0$\, and \,$(\cdot)_1$\, be the ``inverses" of \,$\langle \cdot,\cdot\rangle$: for any $w\in\n,\ \,\langle(w)_0,(w)_1\rangle=w$.
\end{itemize}\vspace*{-0.5mm}
\end{definition}

Classically one has:

\begin{claim}\label{p13 claim1}
The function \,$\langle\cdot,\cdot\rangle$\, and its inverses \,$(\cdot)_0$,\,$(\cdot)_1$
 are polynomial time computable.
 \end{claim}

 \begin{proof} To answer a referee's request, we justify the second assertion.

 Let \,$s\in\n$. We must find the integer $a$ such that \ \ $a(a+1)\leq 2s<(a+1)(a+2)$ \ because if \,$\displaystyle \Delta=s-\frac{(a(a+1))}{2}$, then one has \ $(s)_0=a-\Delta$ \,and\, $(s)_1=\Delta$.

 It takes quadratic time to get the ``square root" of $2s$: the integer $\alpha$ such that\\  $\alpha^2\leq 2s<(\alpha+1)^2$. Then either $\alpha$ or $\alpha-1$ is the expected $a$.
 \end{proof}

 We derive the coding of triples:

\begin{definition}\label{def triple}
For \,$u,v,w\in\n$, \,let \,$\langle u,v,w\rangle=\langle\langle u,v\rangle,w\rangle$.
\end{definition}

\begin{claim}\label{p13 claim2}
For \,$u,v,w\in\n$, \,$\langle u,v,w\rangle\,\leq\,8(u+v+w)^4$.
\end{claim}

\begin{proof} \ If \,$u+v=k$, then \,$\langle u,v\rangle\,<\,\langle k+1,0\rangle$. \ Hence
\begin{equation}\label{e5}
\langle u,v\rangle\,\leq\,\frac{(k+1)(k+2)}{2}-1\,\leq\,2k^2\,=2(u+v)^2
\end{equation}
We deduce
\begin{align*}
\langle u,v,w\rangle \leq\ \langle\langle u,v\rangle,w\rangle & \leq\ \langle 2(u+v)^2,w\rangle\espa \text{\ (by~\eqref{e5})}\\
       & \leq\ 2(2(u+v)^2+w)^2\esp \text{(by~\eqref{e5})}\\
       & \leq\ 8(u+v+w)^4.
\end{align*}

\vspace*{-6mm}
\end{proof}

In order to deal with finite sequences of arbitrary length of integers, we follow one referee's suggestion: writing successively the integers under binary representation while separating them with a new symbol. To keep binary sequences, we replace the symbol 0 by 00, 1 by 11 and the new symbol by 01. We thus consider the following:

\begin{definition}\label{code string}
Let $Seq$ be the predicate on $\n$ defined as follows: for $s\in\n$,
\begin{enumerate}[topsep=-1mm,itemsep=-.2ex]
\item $Seq(s)$ \ iff
 $\ \ s=\sum_{i<2t}\varepsilon_i2^i, \text{ for }t>1 \text{  such that}$
\begin{enumerate}[label=\alph*),topsep=-1mm,itemsep=-.2ex]
\item $\varepsilon_0=\varepsilon_1$
\item $\varepsilon_{2t-2}=0, \ \varepsilon_{2t-1}=1$
\item $\text{ for any }j<t-2,\text{ if }\varepsilon_{2j}=0, \varepsilon_{2j+1}=1,\text{ then }\varepsilon_{2j+2}=\varepsilon_{2j+3}.$
\end{enumerate}
\item Let \ $S(s)=\{i<t:\varepsilon_{2i}=0,\ \varepsilon_{2i+1}=1\}\ \text{ and }\ l(s)=|S(s)|$ \,(the cardinality of \,$S(s)$). If \ $(i_j)_{j<l(s)}$ is an increasing enumeration of \,$S(s)$, then we set
\begin{itemize}[topsep=-1mm,itemsep=-.2ex]
\item $s(0)=\sum_{n<i_0}\varepsilon_{2n}2^{i_0-1-n}$ and
\item for $1\leq j<l(s)$, \ $s(j)=\sum_{i_{j-1}<n<i_j}\varepsilon_{2n}2^{(i_{j}-1)-n}$.
\end{itemize}
$(s(j))_{j<l(s)}$ \,is the sequence of integers encoded in $s$.
 \end{enumerate}
\end{definition}

We note the following:

\begin{fact}\label{size_sequence}~
Let $\mu,l\in\n$ \,and\, $\mathbf{a}=(a_i)_{i<l}$ \,be a sequence of integers such that for any \,$i<l,\ a_i\leq \mu$. Then by the previous definition, we can encode \,$\mathbf{a}$\, in $s\in\n$ such that $Seq(s)$ holds, $l(s)=l$, for any $i<l(s),\ s(i)=a_i$ \ and \ $2l\leq |s|\leq 2l(|\mu|+1)$.
\end{fact}

One easily checks:

 \begin{claim}\label{seq poly}
 The predicate $Seq$ can be checked in polynomial time and the functions \,$s\to l(s)$\, and\, $(s,i)\to s(i)$, \,for $i<l(s)$, are computable in polynomial time.
 \end{claim}


\section{The tree associated with the computation of $\pmb A$}

This section is greatly inspired from Tourlakis' exposition~\cite[Section 2.4.4]{tou} of the fact that $graph(A)=\{(u,v,w):A(u,v)=w\}$ is primitive recursive (where \,$A$\, is defined with different initial conditions). In order to deal with the inverse $\alpha$ of the function $u\mapsto A(u,u)$, we shall consider in addition the predicate \,$A(u,v)<w$. To control the size of Tourlakis' type tree witnessing \,$A_k(n)<m$, it will be helpful to add new leaves.

\medskip
Let us first note that for $u,v\geq 1$, since \ $A_u(v)=A_{u-1}(A_u(v-1))$, one gets the following equivalences:
\begin{equation}
A_u(v)=w\ \Leftrightarrow\ \text{there exists } \,w'>0\ \text{ such that } \left\{\begin{array}{l}
A_u(v-1)=w' \ \text{ and }\\
A_{u-1}(w')=w
\end{array}\right.
\end{equation}

\begin{equation}
A_u(v)<w\ \Leftrightarrow\ \text{there exists }\,w'>0\ \text{ such that } \left\{\begin{array}{l}
A_u(v-1)=w'\ \text{ and }\\
A_{u-1}(w')<w
\end{array}\right.
\end{equation}

($w'>0$ because for any $x,y\in\n,\ A_x(y)\geq A_0(0)=1>0$)

\medskip
As in~\cite{tou}, one can thus unroll a labeled binary tree witnessing the fact that \,$A_k(n)<m$. We shall restrict to the case $k\geq 4$ and $n\geq 3$. For some unique sequence \,$(w_i)_i$, \,the labels of the nodes in the following tree are true statements 
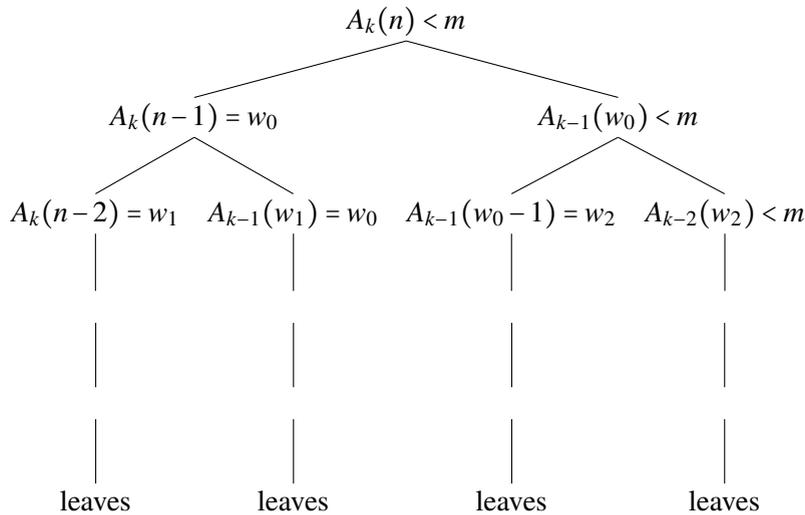
\vspace*{4mm}

\centerline{\synttree[$A_k(n)<m$[$A_k(n-1)=w_0 $ [$A_k(n-2)=w_1 $[\esp[\esp[leaves]]] ][$A_{k-1}(w_1)=w_0 $ \branchheight{1in}[\esp[\esp[leaves]]]]][$A_{k-1}(w_0)<m$ [$A_{k-1}(w_0-1)=w_2$[\esp[\esp[leaves]]]][$A_{k-2}(w_2)<m$[\esp[\esp[leaves]]]] ]]}

\bigskip
Let us describe the structure of the tree and the labeling of nodes:
\begin{itemize}[topsep=0mm,itemsep=-.2ex]
\item the root is labeled $A_k(n)<m$.
\item A node which is not a leaf is labeled
\begin{itemize}[topsep=-1mm,itemsep=-.2ex]
\item either by \,$A_u(v)=w$ \,for $u,v\geq 1$, and admits for a unique\, $w'>0$, a left son labeled \\$A_u(v-1)=w'$\, and a right son labeled \,$A_{u-1}(w')=w$,
\item or by \,$A_u(v)<m$\, for $u\geq 4,\,v\geq 1$, and admits for a unique\, $w'>0$, a left son labeled \\$A_u(v-1)=w'$\, and a right son labeled \,$A_{u-1}(w')<m$.
\end{itemize}
\item A node is a leaf if it is labeled
\begin{itemize}[topsep=-1mm,itemsep=-.2ex]
\item either by \,$A_0(v)=2^v$,
\item or by \,$A_u(0)=1$,
\item or by \,$A_u(v)<m$\, with \,$u\leq 3$\, or \,$v=0$.
\end{itemize}
\end{itemize}

\begin{claim}\label{claim p17} Let us consider a binary tree  witnessing \,$A_k(n)<m$, \,for \,$k\geq 4,\ n\geq 3$.
\begin{enumerate}[label=(\alph*),topsep=-1mm,itemsep=-.2ex]
\item
\begin{enumerate}[label=\arabic*.,leftmargin=*,topsep=1mm,itemsep=-.2ex]
\item If a node in the tree is labeled $A_u(v)=w$, then \,$u,v,w<\,log^{(4)}(m)$.
\item If it is labeled \,$A_u(v)<m$, then \,$3\leq u,v<\,log^{(4)}(m)$.
\end{enumerate}
\item If a node labeled \,$A_u(v)=w$\, or \,$A_u(v)<m$\, admits a son labeled \,$A_{u'}(v')=w'$ \,or \,$A_{u'}(v')<m$, \,then \ $(u',v')<_{lex}(u,v)$.
\end{enumerate}
\end{claim}


\begin{proof}  \ (a) \,We argue by induction on the level of the node:\smallskip

\textbf{Level\, 0: } The root is labeled \,$A_k(n)<m$. By hypothesis, $k,n\geq 3$. Hence
\begin{align*}exp^{(4)}(n)&\leq A_3(n)\leq A_k(n)<m\espa\text{(by Claim~\ref{claim p4})}\\
exp^{(4)}(k)&\leq A_k(3)\leq A_k(n)<m\espa\text{(by Claim~\ref{claim1 p5})}
\end{align*}
Therefore both \,$k,n<\,log^{(4)}(m)$ \,and (a)\,2. holds.\smallskip

\textbf{Level\, 1: } there are two nodes of level 1: the left node is labeled \,$A_k(n-1)=w_0$\, and the right node \,$A_{k-1}(w_0)<m$.\nl
$k\geq 4$\, implies \,$k-1\geq 3$. \ Also\,
$A_k(n-1)\geq A_0(2)=2^2$ \,implies \,$w_0\geq 4\geq 3$.\nl
It remains to check \,$w_0<\,log^{(4)}(m)$. \,By Claim~\ref{claim p4}, \,$k-1\,\geq\, 3$ \,and\, $w_0\,\geq \,3$\, imply\nl
\centerline{$exp^{(4)}(w_0)\leq A_3(w_0)\leq A_{k-1}(w_0)<m$.}\nl
Hence (a)\,1. and (a)\,2. hold at level 1.\smallskip

\textbf{level \,$\pmb{r+1}$ \,with\, $\pmb{r\geq 1}$:} \ We assume the properties hold for the nodes at level $r$ and we check that it is also true for their sons.
\begin{itemize}[topsep=1mm,itemsep=-.2ex]
\item Let thus the node of level $r$ be labeled \,$A_u(v)=w$\, with \,$u,v,w<\,log^{(4)}(m)$.
\begin{itemize}[topsep=1mm,itemsep=-.2ex]
\item Its left son is labeled \,$A_u(v-1)=w'$,
\item its right son is labeled \,$A_{u-1}(w')=w$.
\end{itemize}
Since \,$w'=A_u(v-1)<A_u(v)=w$, \ (a)\,1. holds for both sons.
\item Let now the node of level \,$r$\, be labeled \,$A_u(v)<m$\, with \,$3\leq u,v<\,log^{(4)}(m)$. Since it is not a leaf, $u\geq 4$.
\begin{itemize}[topsep=-1mm,itemsep=-.2ex]
\item Its left son is labeled \,$A_u(v-1)=w'$,
\item its right son is labeled \,$A_{u-1}(w')<m$.
\end{itemize}
As for level 1, $v-1\geq 2$\, implies \,$w'=A_u(v-1)\geq A_0(2)\geq 3$. Also $u-1\geq 3$ and \,$A_{u-1}(w')<m$ \,give \,$w'<\,log^{(4)}(m)$. We thus deduce that (a)\,1. holds for the left son and (a)\,2. for the right one.
\end{itemize}\medskip

(b) \,Keeping the notation of the claim, we simply note  $(u',v')=
\begin{cases}
(u-1,w') \text{\ \ or}\\
(u,v-1).
\end{cases}$

Hence \,$(u',v')<_{lex}(u,v)$.
\end{proof}

\begin{remark}~\label{leaf4}
\begin{itemize}[topsep=-1mm,itemsep=-.2ex]
\item (a)\,2.  implies that the last type of leaf labeled \ $A_u(v)<m$ \,with \,$u\leq 3$\, or \,$v=0$, is necessarily of the form  \,$A_3(v)<m$\, for $v\geq 3$.
\item We also deduce from this claim that the binary tree witnessing \,$A_k(n)<m$, \,for $k\geq 4,$ $\,n\geq 3$, \,has height at most $\,(log^{(4)}(m))^2$.
\end{itemize}
 \end{remark}

We do not know whether different nodes in the tree may have the same label. This made the exposition a bit more tedious. We now focus on labels  (which are true statements) occuring in the binary tree witnessing \,$A_k(n)<m$ \,and encode this set.\nl
let us first note that if \,$A_u(v)=w$, \,then \,$w\geq A_0(0)=1$. Hence we shall represent the label \,$A_u(v)=w$\, by the integer \,$\langle u,v,w\rangle$\, and the label \,$A_u(v)<m$ \,by the integer \,$\langle u,v,0\rangle$ \,(this will reduce the size of the encoding). We  identify the label with its code.
Let us introduce the following notation:
\begin{definition}\label{lex}
If $x'=\langle u',v',w'\rangle$\, and \,$ x=\langle u,v,w\rangle$, then \, \,\nl
\centerline{$x'<_{lex}^3 x$ \ssi \,$(u',v',w')<_{lex}(u,v,w)$.}
\end{definition}

By Claim~\ref{claim p17}\,(b), if $\langle u',v',w'\rangle$\,  labels  the son of a node labeled  \,$\langle u,v,w\rangle$, \,then \,$\langle u',v',w'\rangle\,<_{lex}^3\,\langle u,v,w\rangle$. \,This motivates the following:

\begin{claim}\label{claim p20}
Let \,$\pmb{a}=(a_i)_{i<l}$\, enumerate, according to increasing $<_{lex}^3$ order, \,all labels occuring in the tree witnessing \,$A_k(n)<m$, \,for \,$k\geq 4,\,n\geq 3$. Then
\begin{enumerate}[label=(\alph*),topsep=-1mm,itemsep=-.2ex]
\item \,$a_{l-1}=\langle k,n,0\rangle$ \,and for any \,$i<l$,
\begin{itemize}[topsep=-1mm,itemsep=-.2ex]
\item either ($a_i$ labels a leaf) \,$a_i=\begin{cases}
\langle 0,v,2^v\rangle \text{\ \ \  or}\\
\langle v,0,1\rangle \text{\ \ \ \ or}\\
\langle 3,v,0\rangle,
\end{cases}$
\item or \,$a_i=\langle u,v,w\rangle$ \,and there exist \,$j,j'<i,\ w'>0$ \,such that\\ \,$a_j=\langle u,v-1,w'\rangle$ \,and\, $a_{j'}=\langle u-1,w',w\rangle$.
\end{itemize}
\item \,$2\leq l\leq (log^{(4)}(m))^3$ \,and for each $i<l,\ \ a_i<6^4(log^{(4)}(m))^4$.
\end{enumerate}
\end{claim}

\begin{proof} (a) holds by definition of the labeled tree,  Claim~\ref{claim p17}\,(b) and Remark~\ref{leaf4}.\nl
(b) By Claim~\ref{claim p17}\,(a), if $\langle u,v,w\rangle$ \,labels a node, then \,$u,v,w<\,log^{(4)}(m)$. Hence \,$l\leq(log^{(4)}(m))^3$. \
By Claim~\ref{p13 claim2}, \,$\langle u,v,w\rangle\leq 8(u+v+w)^4$. \ Hence \nl \espa\espa $\langle u,v,w\rangle\,<\,2^33^4(log^{(4)}(m))^4\,\leq\,6^4(log^{(4)}(m))^4$.
\end{proof}\medskip

To reduce the time of computation, instead of checking for several $v$'s whether \,$A_3(v)<m$ \,(to recognize a leaf), we shall compute once \,$r_m=Inv_{A_3}(m)$ \,and then check \,$v<r_m$, \,for the different $v$'s. We thus set:

\begin{definition}\label{def comput}
Let us consider the predicate \,$Comput_<$\, defined as follows: for \,$s,k,n,r\in\n$, \,
\begin{equation*}
\begin{split}
Comput_< (s,k,n,r) \ \ \ \text{iff}\ \ \ & Seq(s)\land s(l(s)-1)=\langle k,n,0\rangle\land\ \forall\,i<l(s) \\
&\Big[\exists v\big(s(i)=\langle 0,v,2^v\rangle\,\lor\,(s(i)=\langle v,0,1\rangle\,\lor\,\\[-2pt]
& \quad \makebox[4.3cm]{}(s(i)=\langle 3,v,0\rangle\,\land\,v<r)\big)\Big]\,\lor\\
&\Big[\exists u,v,w\ \exists\,w'>0\ \exists j,j'<i\ \big(\,s(i)=\langle u,v,w\rangle\ \land\ \\[-2pt]
& \quad \makebox[1.7cm]{} s(j)=\langle
 u,v-1,w'\rangle\,\land\,s(j')=\langle u-1,w',w\rangle\big)\Big].
\end{split}
\end{equation*}
\end{definition}

We obtain:

\begin{claim}\label{claim p21}
There exists \,$C\in\n$\, such that for all \,$k\geq 4,\,n\geq 3,\,m\geq 0$\, if \,$A_k(n)<m$\, and \,$r_m=Inv_{A_3}(m)$, then there is \,$s\leq C\,log^{(2)}(m)$ \,such that \,$Comput_<(s,k,n,r_m)$ \,holds.
\end{claim}

\begin{proof}
Let $\pmb{a}=(a_i)_{i<l}$\, be the sequence of Claim~\ref{claim p20} enumerating the different labels occuring in the tree witnessing \,$A_k(n)<m$. By (b) of this claim, if $\mu=(6\,log^{(4)}(m))^4$, then for any $i<l$, $a_i<\mu$.
By Fact~\ref{size_sequence}, let \,$s\in\n$\, encode \,$\pmb{a}$\, and satisfy \ $|s|\leq 2l(|\mu|+1)$.

\medskip
- We first note that \,$Comput(s,k,n,r_m)$ \,holds: if for some \,$i<l,\ s(i)=a_i=\langle 3,v,0\rangle$, \,then this implies \,$A_3(v)<m$\,
and hence \,$v<r_m$. \medskip \\
 \indent - It remains to bound \,$s$. By Claim~\ref{claim p20}\,(b), one has \,$l\leq(log^{(4)}(m))^3$. Since
    $|s|\leq 2l(|\mu|+1)$ \,for \,$\mu=(6\,log^{(4)}(m))^4$, applying \ $s<2^{|s|}$\, and \,$2^{|\mu|}\leq 2\mu$, we obtain
      $s\leq 2^{2l(|\mu|+1)}\leq (2\mu)^{2l}2^{2l}$.

\medskip\noindent Hence \ $s\leq (6log^{(4)}(m))^{8(log^{(4)}(m))^3}\cdot 2^{2(log^{(4)}(m))^3}$.
 There exists \,$K,\,K'$\, (independent of $m$) such that
 \begin{equation*}
 \begin{split}
(6log^{(4)}(m))^{8(log^{(4)}(m))^3}\cdot 2^{2(log^{(4)}(m))^3}& \leq\ K2^{(log^{(4)}(m))^4}\\& \leq\ KK'2^{log^{(3)}(m)}\\& \leq\ 2KK'log^{(2)}(m).
\end{split}
\end{equation*}
(We use the fact that if $f(m)\leq g(m)$ almost everywhere, then there is $\theta$ such that $f(m)\leq g(m)+\theta$ for all $m$, and hence \ $2^{f(m)}\leq 2^\theta 2^{g(m)}$ for all $m$).

Therefore \,$s\leq 2KK'log^{(2)}(m)$.
\end{proof}

Conversely, one obtains:

\begin{claim}\label{claim p23}
Let \,$m\in\n$\, and \,$r_m=Inv_{A_3}(m)$. \,If \,$Comput_<(s,k,n,r_m)$ \,holds for some \,$s,k,n\in\n$, \,then \,$A_k(n)<m$.
\end{claim}

\begin{proof}
 We assume \,$Comput_<(s,k,n,r_m)$\, is satisfied and we check by induction on \,$i<l(s)$\, that
\begin{enumerate}[label=(\alph*),topsep=-1mm,itemsep=-.2ex]
\item if \,$s(i)=\langle u,v,w\rangle$\, with \,$w>0$, \,then \,$A_u(v)=w$.
\eject
\item if \,$s(i)=\langle u,v,0\rangle$, \,then \,$A_u(v)<m$.
\end{enumerate}

\smallskip
Let us note that by definition, $\,l(s)\geq 2$.

\medskip
- Let $i=0$. \,By definition of \,$Comput_<$, \,$s(i)$\, is necessarily of a ``leaf type". That is
\begin{itemize}[topsep=-.8mm,itemsep=-.3ex]
\item either \,$s(i)=\langle 0,v,2^v\rangle$ \,or\, $\langle u,0,1\rangle$, and by definition of the function $A$, (a) is satisfied,
\item or \,$s(i)=\langle 3,v,0\rangle\,$ and \,$Comput(s,k,n,r_m)$\, implies \,$v<r_m$. \,Therefore one has \,$A_3(v)<m$\, and (b) holds.
\end{itemize}

\medskip
- Let now $i>0$. We assume that for any \,$j<i$, according to the nature of \,$s(j)$, (a) or (b) holds for $s(j)$.

If \,$s(i)$\, is of the ``leaf type", then one argues as for $i=0$. Otherwise \,$s(i)=\langle u,v,w\rangle$\, and there exist \,$j,\,j'<i$\, and \,$w'>0$\, such that we have \,$s(j)=\langle u,v-1,w'\rangle$ \,and\, $s(j')=\langle u-1,w',w\rangle$.\\
By induction hypothesis,
\medskip \\
 $\begin{array}[t]{l}
{\scriptstyle{\bullet}} \ \ A_u(v-1)=w' \text{ \ (the fact that }w'>0\text{ is important)}\\
{\scriptstyle{\bullet}} \ \ \text{and }\begin{cases}
\text{ if }w=0,\ \,A_{u-1}(w')<m,\\
\text{ if }w>0,\ \,A_{u-1}(w')=w.
\end{cases}
\end{array}$ \medskip \nl
We thus deduce \ $A_u(v)=A_{u-1}(A_u(v-1))=A_{u-1}(w')$. Hence according to whether \,$w=0$\, or not, we conclude that (a) or (b) holds for \,$s(i)=\langle u,v,w\rangle$.

\medskip
Hence by (b) applied to $i=l(s-1)$ and $s(l(s)-1)=\langle k,n,0\rangle$, we derive \ $A_k(n)<m$.
\end{proof}

Combining Claims~\ref{claim p21} and \ref{claim p23}, we obtain:

\begin{lemma}\label{lemma alpha}
There is \,$C\geq 1$ \,such that for any \,$k\geq 4,\,n\geq 3,\ m\in\n$, \,if \,$r_m=Inv_{A_3}(m)$, \,then the following equivalence holds:\nl
\centerline{$A_k(n)<m\ssi\exists s\leq C\,log^{(2)}(m)\ \,Comput_<(s,k,n,r_m)$.}
\end{lemma}

\section{Computation time}

We first estimate the complexity of \,$Comput_<$:

\begin{claim}\label{claim p25}
There exist \,$B,t\in\n$\, such that the predicate \,``$k\geq 4\,\land\,n\geq 3\ \land\,$\linebreak $Comput_<(s,k,n,r)$" \,can be checked in at most \,$B(max(|s|,|k|,|n|,|r|))^t$ \,steps.
\end{claim}

\begin{proof} \ This is a consequence of Claim~\ref{seq poly} about the complexity of $Seq$, the definition of $Comput_<$ (Definition~\ref{def comput}) and the fact that \,$l(s)\leq |s|$ \,(see Fact~\ref{size_sequence}).\end{proof}\medskip

Our goal is now to obtain:

\begin{lemma}\label{lemma beta}
There is \,$D\in\n$\, such that the predicate \,``$k\geq 4\,\land\,n\geq 3\,\land\,A_k(n)<m$" \,can be checked in at most \,$D\,\max(|k|,|n|,|m|)$\, steps.
\end{lemma}

\begin{proof} \ Let \,$C\geq 1$\, be the constant mentioned in Lemma~\ref{lemma alpha}.\medskip

\textbf{The algorithm which decides the predicate \ }``$\pmb{ k\geq 4\,\land\,n\geq 3\,\land\,A_k(n)<m}$":
\begin{enumerate}[label=(\arabic*),topsep=1mm,itemsep=-.2ex]
\item We check \,$k\geq 4,\,n\geq 3$\, and then \,$k,n\leq log^{(2)}(m)$,
\item we compute \,$r_m=Inv_{A_3}(m)$\, and check \,$r_m>3$.
\item If these previous steps have been successfully completed, we try all \,$s\leq C\,log^{(2)}(m)$ \,to obtain \,$Comput_<(s,k,n,r_m)$.
If we fail to obtain such an $s$ or  to satisfy steps (1) and (2), then we output ``No". Otherwise it is ``yes".
\end{enumerate}

We note that \,$r_m\leq 3$\, implies \,$A_3(3)\geq m$\, and hence \,$A_k(n)\geq A_3(3)\geq m$. Hence as the requirement \,``$k,n\leq log^{(2)}(m)$", the condition \,$r_m>3$\, can be harmlessly added in the definition of the algorithm. Their role is only to reduce the running time of the algorithm.\nl
By Lemma~\ref{lemma alpha}, the algorithm is correct.\medskip

\textbf{Running time}:
\begin{enumerate}[label=(\arabic*),topsep=-1mm,itemsep=-.2ex]
\item By Lemma~\ref{lemma2 p5}, step (1) requires at most \,$\+O(\max(|k|,|n|,|m|))$\, steps
\item By Lemma~\ref{time invk}, step (2) needs at most \,$\+O(|m|)$\, steps.
\item If \,$r_m>3$, then \,$r_m-1\geq 3$\, and \,$A_3(r_m-1)<m$\, implies because of Claim~\ref{claim p4}  \,$r_m\leq log^{(4)}(m)$. \,We thus have \,$s,k,n,r_m\,\leq\,C\cdot log^{(2)}(m)$\, and hence \nl\centerline{$
|s|,|k|,|n|,|r_m|\,\leq\,2C\cdot log^{(3)}(m)$}\nl
(because if \,$d,log(v)\geq 1$, then \,$u\leq dv$\, implies \,$|u|\leq 2d\,log(v)$).\nl
By Claim~\ref{claim p25}, there are \,$B,t\in\n$\, such that, for each $s\in\n$, checking\\  $Comput_<(s,k,n,r_m)$\, takes at most \,$B(\max(|s|,|k|,|n|,|r_m|))^t$\, steps.\nl
Hence checking for all $s\leq C\,log^{(2)}(m)$, whether \,$Comput_<(s,k,n,r_m)$\, holds, requires at most \,$T\,=\,B2^tC^{t+1}log^{(2)}(m)(log^{(3)}(m))^t$ \,steps.\nl
There is $K\in\n$ (independent of $m$) such that  \,$T\,\leq\,K(log^{(2)}(m))^2$

Using  \,$(log(r))^2\leq 4r$, for \,$r\geq 4$, we deduce \,$T\,\leq\,4K\,log(m)\,\leq\,4K|m|$.
\end{enumerate}

Lemma~\ref{lemma beta} follows from the time estimates of (1),(2) and (3).
\end{proof}

It suffices now to remove the hypothesis \,``$k\geq 4\,\land\,n\geq 3$".

\begin{lemma}\label{lemma p27}
There is a constant $D\in\n$ such that for any $k,n,m\in\n$, the predicate ``$A_k(n)<m$" can be checked in at most \,$D\,\max(|k|,|n|,|m|)$\, steps.
\end{lemma}

\begin{proof} \ Let $\begin{array}[t]{lll}
U(k,n,m)&\text{ iff } &k\geq 4\,\land\,n\geq 3\,\land\,A_k(n)<m,\nl
V(n,k,m)&\text{ iff } &k\leq 3\,\land\,A_k(n)<m,\nl
W(k,n,m)&\text{ iff } &n\leq 2\,\land\,A_k(n)<m.
\end{array}$

\medskip
Then \ \ $A_k(n)<m\ssi U(k,n,m)\,\lor\,V(k,n,m)\,\lor\,W(k,n,m)$.\nl
- By Lemma~\ref{lemma beta}, \,$U(k,n,m)$\, can be checked in \,$\+O(\max(|k|,|n|,|m|))$\, steps.\nl
- For any $k,n,m$, one has the equivalence: \ $V(k,n,m)\ \Leftrightarrow \ k\leq 3\,\land\,Inv_{A_k}(m)>n$\\
Hence by Lemma~\ref{time invk}, \,$V(k,n,m)$ \,can be checked in \,$\+O(\max(|k|,|n|,|m|))$\, steps.\nl
- By Fact~\ref{fact p3}, for any \,$k,n,m$, one has the equivalences:
\begin{align*}
W(k,n,m)&\Leftrightarrow\  (n=0\,\land\,A_k(0)<m)\,\lor \,(n=1\,\land\,A_k(1)<m)\,\lor \,(n=2\,\land\,A_k(2)<m)\\
&\Leftrightarrow\ (n=0\,\land\,m>1)\,\lor\,(n=1\,\land\,m>2)\,\lor\,(n=2\,\land\,m>4).
\end{align*}
Hence \,$W(k,n,m)$\, can also be verified in \,$\+O(\max(|k|,|n|,|m|))$\, steps.
\end{proof}

Let \,$Graph(A)=\{(k,n,m)\in\n^3\,:\,A_k(n)=m\}.$ We deduce:

\begin{proposition}\label{prop graph}
The predicate ``$(k,n,m)\in Graph(A)$" is checkable in linear time.
\end{proposition}

\begin{proof} \ For any \,$k,n,m\in\n$, \ $A_k(n)=m\ssi A_k(n)<m+1\,\land \,\neg(A_k(n)<m)$.
\end{proof}

Let us recall that the function \,$Ack$\, is such that, for any \,$k\in\n$, \,$Ack(k)=A(k,k)$\, and \,$\alpha$\, is its inverse \,$Inv_{Ack}$ (definitions~\ref{def-A}(c) and~\ref{def-ack}).\nl
To obtain the fact that \,$Graph(A)$\, is checkable in linear time, we could have as in Tourlakis' book, considered the predicate \,``$A_k(n)=m$"\, in place of \,``$A_k(n)<m$". But to prove that $\alpha$ itself, can be computed in linear time, it seemed to us that the use of the predicate \,``$A_k(n)<m$" was necessary. It is not the case for some approximations; for instance \,$\alpha'\!:n\mapsto\alpha(log^{(2)}(n))$ satisfies for any \,$n\in\n,\linebreak\,\ 0\leq \alpha(n)-\alpha'(n)\leq 2$ and the fact that it is computable in linear time can be deduced from the fact that  \,$Graph(A)$\, is checkable in linear time.


\begin{proposition}\label{alpha linear}
The function \,$\alpha$\, is computable in linear time.
\end{proposition}

\begin{proof} \ Let \,$C\geq 1$\, be the constant of Lemma~\ref{lemma alpha}. We now propose an algorithm which on input $m\in\n$, outputs \,$\alpha(m)$.\medskip

\textbf{The algorithm}: \ let \,$m\in\n$.
\begin{enumerate}[label=(\arabic*),topsep=-1mm,itemsep=-.2ex]
\item For each \,$k\leq 3$, we compute \,$\rho_k=Inv_{A_k}(m)$. If there is \,$k\leq 3$\, such that \,$\rho_k\leq k$, then we output the least such $k$. Otherwise we go to step 2.
\item We compute \,$log^{(4)}(m)$ (we shall see that necessarily it is greater or equal to 4). For each \,$j$\, such that \,$4\leq j\leq log^{(4)}(m)$, we test all \,$s\leq C\,log^{(2)}(m)$ \,to obtain \,$Comput_<(s,j,j,\rho_3)$.\nl
We output the least \,$j_0\geq 4$\, for which we fail to find such an $s\leq C\,log^{(2)}(m)$.
\end{enumerate}

\medskip
\textbf{Validity of the algorithm}:
\begin{enumerate}[label=(\arabic*),topsep=-1mm,itemsep=-.2ex]
\item If we stopped after step (1) and $k_0\leq 3$ is least such that \,$\rho_{k_0}\leq k_0$, then\\ $k_0\geq Inv_{A_{k_0}}(m)$ \,implies \,
\begin{equation}\label{asterisk}
A_{k_0}(k_0)\geq m.
\end{equation}
\begin{itemize}[topsep=-1mm,itemsep=-.2ex]
\item If \,$k_0=0$, then by \eqref{asterisk} \,$\alpha(m)=0$,
\item otherwise, by definition, \,$\rho_{k_0-1}>k_0-1$. Hence \,$\rho_{k_0-1}-1\geq k_0-1$ \,and we deduce\vspace*{-2mm}
\begin{equation}\label{2asterisk}
A_{k_0-1}(k_0-1)\,\leq\,A_{k_0-1}(\rho_{k_0-1}-1)\,<\,m. \vspace*{-3mm}
\end{equation}
\eqref{asterisk}+\eqref{2asterisk} give \,$\alpha(m)=k_0$.
\end{itemize}
\item We thus assume now that for any $k\leq 3,\ \rho_k>k$. \,Hence \,$\rho_3>3$\, and \,$A_3(3)<m$. \,By Fact~\ref{case3}\vspace*{-2mm}
\begin{equation}\label{un}
log^{(4)}(m)>3.\vspace*{-3mm}
\end{equation}
 Hence the following inequalities hold:
\begin{align*}
A(log^{(4)}(m),log^{(4)}(m))\ & \geq\  A_3(log^{(4)}(m))\ \espa\text{(by~\eqref{un})}\\
                            & \geq\  exp^{(4)}(log^{(4)}(m))\esp\text{(by~\eqref{un} and Claim~\ref{claim p4})}\\
                            & \geq\  m.
\end{align*}
Hence \ $4\leq\alpha(m)\leq log^{(4)}(m)$. \,Let \,$j_0=\alpha(m)$. \,Then for any \,$i<j_0$, one has \,$A(i,i)<m$. \,By Lemma~\ref{lemma alpha}, we must succeed in finding \,$s\leq C\,log^{(2)}(m)$\, such that \,$Comput_<(s,i,i,\rho_3)$ \,and we must fail in finding one such \,$s$\, satisfying  \,$Comput_<(s,j_0,j_0,\rho_3)$\, since \,$A(j_0,j_0)\geq  m$.\nl
\end{enumerate}
Hence the algorithm outputs \,$\alpha(m)$.\medskip

\textbf{Running time of the algorithm}:
\begin{enumerate}[label=(\arabic*),topsep=-1mm,itemsep=-.2ex]
\item By Lemma~\ref{time invk}, step (1) takes \,$\+O(|m|)$\, steps.
\item We know \,$\,\rho_3>3$. By Claim~\ref{claim p4}, \,$\rho_3-1\geq 3$\, and \,$A_3(\rho_3-1)<m$ \,imply\, $\rho_3-1<\,log^{(4)}(m)$\, and\, $\rho_3\leq\,log^{(4)}(m)$.\nl
For each \,$4\leq j\leq log^{(4)}(m)$\, and for all \,$s\leq C\,log^{(2)}(m)$, \,we check\\
$Comput_<(s,j,j,\rho_3)$.\\
Since \,$s,j,\rho_3\,\leq\,C\,log^{(2)}(m)$, \,as in the proof of Lemma~\ref{lemma beta}, we obtain\nl \centerline{\,$|s|,|j|,|\rho_3|\,\leq\,2C\,log^{(3)}(m)$.}\nl
By Lemma~\ref{claim p25}, there are \,$B,t\in\n$\, such that \,$Comput_<(s,j,j,\rho_3)$\, can be checked in at most \,$B(\max(|s|,|j|,|\rho_3|))^t$\, steps.\nl
We deduce that step (3) can be completed in at most \nl
\centerline{$T\,=\,B2^tC^{t+1}log^{(2)}(m)\,log^{(4)}(m)\,(log^{(3)}(m))^t$\, steps.}\nl
There is $\,K\in\n$\, (independent of $m$) such that \,$T\leq K\,log(m)\leq K\,|m|$.
\end{enumerate}
Hence steps (1) and (2) take time \,$\+O(|m|)$.\end{proof}\medskip

These methods can be applied to the different two argument  inverse Ackermann functions proposed  in~\cite{tar,cha,seidel,seidel2}.

\end{document}